\documentclass[preprint]{imsart}
\RequirePackage[OT1]{fontenc}
\RequirePackage{amsthm,amsmath,amssymb,enumerate}
\RequirePackage[numbers]{natbib}
\RequirePackage[colorlinks,citecolor=blue,urlcolor=blue]{hyperref}

\usepackage{algorithm}
\usepackage{algcompatible}
\usepackage{caption}
\usepackage{subcaption}
\usepackage{graphicx}
\usepackage{geometry}
\usepackage{array}
\usepackage{dsfont}
\usepackage{changes}
\definechangesauthor[name={Guilherme Ost}, color=red]{GO}
\usepackage{soul,color}
\soulregister\cite7
\soulregister\ref7
\soulregister\pageref7

\startlocaldefs
\numberwithin{equation}{section}
\theoremstyle{plain}

\newtheorem{proposition}{Proposition}
\theoremstyle{definition}

\endlocaldefs

\usepackage{comment}
\usepackage{algpseudocode}
\usepackage{tabularx}

\algnewcommand{\LeftComment}[1]{\Statex \(\triangleright\) #1}

\usepackage{tikz}
\usetikzlibrary{decorations.pathreplacing,calc}
\newcommand{\tikzmark}[1]{\tikz[overlay,remember picture] \node (#1) {};}

\newcommand*{\SpaceReservedForComments}{1.7cm}%
\newcommand*{\HorizontalOffset}{-0.5em}%
\newcommand*{\VerticalOffset}{0.7ex}%
\newcommand*{\AddNote}[4][]{%
    \begin{tikzpicture}[overlay, remember picture]
        \draw [decoration={brace,amplitude=0.5em},decorate,ultra thick,red, #1]
            ($(#3)+(\HorizontalOffset,-\VerticalOffset)$) --  ($(#2)+(\HorizontalOffset,\VerticalOffset)$)
            node [align=left, text width=\SpaceReservedForComments-1.0em, pos=0.5, anchor=east] {#4};
    \end{tikzpicture}
}%

\makeatletter
    \algrenewcommand\alglinenumber[1]{\tikzmark{\arabic{ALG@line}}\tiny#1:}
\makeatother

\makeatletter
\newcommand{\multiline}[1]{%
  \begin{tabularx}{\dimexpr\linewidth-\ALG@thistlm}[t]{@{}X@{}}
    #1
  \end{tabularx}
}
\makeatother

\begin{document}

\begin{frontmatter}

\title{Event-scheduling algorithms with Kalikow decomposition for simulating potentially infinite neuronal networks
}
\runtitle{}

\begin{aug}

\author{\fnms{T.C.} \snm{Phi}\thanksref{m1}
\ead[label=e2]{cuongphi@unice.com}}
\and
\author{\fnms{A.} \snm{Muzy}\thanksref{m2}\ead[label=e2]{muzy@i3s.unice.fr} 
}\and 
\author{\fnms{P.} \snm{Reynaud-Bouret}\thanksref{m1}\ead[label=e2]{reynaudb@unice.fr}
}

\affiliation{Universit\'{e} C\^{o}te d'Azur, CNRS, LJAD, France  \thanksmark{m1} and Universit\'{e} C\^{o}te d'Azur, CNRS, I3S,France\thanksmark{m2}}

\thankstext{t1}{\today}
\thankstext{m1}{Universit\'{e} C\^{o}te d'Azur, CNRS, LJAD, France}
\thankstext{m2}{Universit\'{e} C\^{o}te d'Azur, CNRS, I3S,France}

\runauthor{T.C. Phi, A. Muzy and P. Reynaud-Bouret}

\end{aug}

\begin{abstract}
Event-scheduling algorithms can compute in continuous time the next occurrence of points (as events) of a counting process based on their current conditional intensity. In particular event-scheduling algorithms can be adapted to perform the simulation of finite neuronal networks activity. These algorithms are based on Ogata's thinning strategy \cite{Oga81}, which always needs to simulate the whole network to access the behaviour of one particular neuron of the network. On the other hand, for discrete time models, theoretical algorithms based on Kalikow decomposition can pick at random influencing neurons and perform a perfect simulation (meaning without approximations) of the behaviour of one given neuron embedded in an infinite network, at every time step. These algorithms are currently not computationally tractable in continuous time. To solve this problem, an event-scheduling algorithm with Kalikow decomposition is proposed here for the sequential simulation of  point processes neuronal models satisfying this decomposition. This new algorithm is applied to infinite neuronal networks whose finite time simulation is a prerequisite to realistic brain modeling.
\end{abstract}



\end{frontmatter}
\newcommand{\pat}[1]{\textcolor{magenta}{#1}}

\maketitle

\section{Introduction}
\label{intro}
Point processes in time are stochastic objects that model efficiently event occurrences with a huge variety of applications:  time of deaths, earthquake occurrences, gene positions on DNA strand, etc. \cite{andersen,VeO82,RBS}). 
    
Most of the time, point processes are multivariate \cite{didelez}  in the sense that either several processes are considered at the same time, or in the sense that one process regroups together all the events of the different processes and marks them by their type. A typical example consists in considering either two processes, one counting the wedding events of a given person and one counting the children birth dates of the same person or only one marked process which regroups all the possible dates of birth or weddings independently  and adds one mark per point, here wedding or birth.
 
Consider now a network of neurons each of them emitting action potentials (spikes). These spike trains can be modeled by a multivariate point process with  a potentially infinite number of marks, each mark representing one  given neuron. The main difference between classical models of multivariate point processes and the ones considered in particular for neuronal networks is the size of the network. A human brain consists in about $10^{11}$ neurons whereas a cockroach contains already about $10^6$ neurons. Therefore the simulation of the whole network is either impossible or a very difficult and computationally intensive task for which  particular tricks depending on the shape of the network or the point processes have to be used  \cite{peters,dassios,Mascart2019}.

Another point of view, which is the one considered here,  is to simulate, not the whole network, but the events of one particular node or neuron, embedded in and interacting with  the whole network. In this sense, one might consider an infinite network.
This is the mathematical point of view considered in a series of papers  \cite{gl1,gl2,orb} and based on Kalikow decomposition \cite{Kalikow} coupled with perfect simulation theoretical algorithms \cite{CFF,FFG}. However these works are suitable in discrete time and only provide a way to decide at each time step if the neuron is spiking or not. They cannot operate in continuous time, i.e. they cannot directly  predict the next event (or spike). Up to our knowledge, there exists only one attempt of using such decomposition in continuous time \cite{HL}, but the corresponding simulation algorithm is purely theoretical in the sense that the corresponding conditional Kalikow decomposition should exist given the whole infinite realization of a multivariate Poisson process, with an infinite number of marks, quantity which is impossible to simulate in practice.

The aim of the present work is to present an algorithm which
\begin{itemize}
\item can operate in continuous time in the sense that it can predict the occurrence of the next event. In this sense, it is an event-scheduling algorithm;
\item can simulate the behavior of one particular neuron embedded in a potentially infinite network without having to simulate the whole network;
\item is based on an unconditional Kalikow decomposition and in this sense, can only work for point processes with this decomposition.
\end{itemize}
In Section 2, we specify the links between event-scheduling algorithms and the classical point process theory. In Section 3, we give the Kalikow decomposition. In Section 4, we present the backward-forward perfect simulation algorithm and prove why it almost surely ends under certain conditions. In Section 5, we provide simulation results and a conclusion is given in Section 6.

\section{Event-scheduling simulation of point processes}
On the one hand, simulation algorithms of multivariate point processes \cite{Oga81} are quite well known in the statistical community but as far as we know quite confidential in the simulation (computer scientist) community. On the other hand, event-scheduling simulation first appeared in the mid-1960s~\cite{tocher} and was formalized as discrete event systems in the mid-1970s~\cite{tms76} to interpret very general  simulation algorithms scheduling ``next events''. A basic event-scheduling algorithm ``jumps'' from one event occurring at a time stamp $t \in \mathds{R}_0^+$ to a next event occurring at a next time stamp $t' \in \mathds{R}_0^+$, with $t'\geq t $. In a discrete event system, the state of the system is considered as changing at times $t,t'$ and conversely unchanging in between \cite{tms18}. In \cite{Mascart2019}, we have written the main equivalence between the point processes simulation algorithms and the discrete event simulation set-up, which led us to a significant improvement in terms of computational time when huge but finite networks are into play. Usual event-scheduling simulation algorithms have been developed considering independently the components (nodes) of a system. Our approach considers new algorithms for activity tracking simulation~\cite{cise-muzy}. The event activity is tracked from active nodes to children (influencees).

Here we just recall the main ingredients that are useful for the sequel.

To define point processes, we need a  filtration or history $(\mathcal{F}_t)_{t\geq 0}$. Most of the time, and this will be the case here, this filtration (or history) is restricted to the internal history of the multivariate process $(\mathcal{F}^{int}_t)_{t\geq 0}$, which means that at time $t-$, i.e. just before time $t$, we only have access to the events that have already occurred in the past strictly before time $t$, in the whole network. The conditional intensity, $\phi_i(t|\mathcal{F}^{int}_{t-})$, of the  point process, representing neuron $i$ gives the instantaneous firing rate, that is the frequency of spikes, given the past contained in $\mathcal{F}^{int}_{t-}$. Let us just mention two very famous examples.

If $\phi_i(t|\mathcal{F}^{int}_{t-})$ is a deterministic constant, say $M$, then the spikes of neuron $i$ form a homogeneous Poisson process with intensity $M$. The occurrence of spikes are completely independent from what occurs elsewhere in the network and from the previous occurrences of spikes of neuron $i$.

If we denote by ${\bf I}$ the set of neurons, we can also envision the following form for the conditional intensity:
\begin{equation}
\label{Hawkesdef}
\phi_i(t|\mathcal{F}^{int}_{t-})= \nu_i + \sum_{j\in {\bf I}} w_{ij} ({\bf Nb}^j_{[t-A,t)}\wedge M).
\end{equation}

This is a particular case of generalized Hawkes processes \cite{bm}. More precisely $\nu_i$ is the spontaneous rate (assumed to be less than the deterministic upper bound $M>1$) of neuron $i$. Then every neuron in the network can excite  neuron $i$: more precisely, one counts the number of spikes that have been produced by neuron $j$ just before $t$, in a window of length $A$, this is $ {\bf Nb}^j_{[t-A,t)}$; we clip it by the upper bound $M$ and modulate its contribution to the intensity by the positive synaptic weight between neuron $i$ and neuron $j$, $w_{ij}$. For instance, if there is only one spike in the whole network just before time $t$, and if this happens on neuron $j$, then the intensity for neuron $i$ becomes $\nu_i+w_{ij}$. The sum over all neurons $j$  mimics the synaptic integration that takes place at neuron $i$. As a renormalization constraint, we assume that 
$\sup_{i\in {\bf I}}\sum_{j\in {\bf I}} w_{ij} <1$.
This ensures in particular that such a process has always a conditional intensity bounded by $2M$.

Hence, starting for instance at time $t$, and given the whole past, one can compute the next event in the network by computing for each node of the network the next event in absence of other spike apparition. To do so, remark that in absence of other spike apparition, the quantity 
$\phi_i(s|\mathcal{F}^{int}_{s-})$ for $s>t$ becomes for instance in the previous example $$\phi^{abs}_i(s,t)=\nu_i + \sum_{j\in {\bf I}} w_{ij} ({\bf Nb}^j_{[s-A,t)}\wedge M), $$
meaning that we do not count the spikes that may occur after $t$ but before $s$. This can be generalized to more general point processes. The main simulation loop is presented in Algorithm~\ref{alg:1}
\begin{algorithm}[H]
    \caption{Classical point process simulation algorithm} 
    \label{algclassic}
    \begin{algorithmic}[1]
    \Statex \LeftComment{With $[t_0,t_1]$ the interval of simulation}
    \State Initialize the family of points ${\bf P}=\emptyset$
    \Statex \LeftComment{Each point is a time $T$ with a mark,  $j_T$, which  is the neuron on which $T$ appears} 
    \State Initialize $t\leftarrow t_0$
    \Repeat
        \For{each neuron $i \in {\bf I}$}
            \State Draw independently an exponential variable $E_i$ with parameter 1 
            \State Apply the inverse transformation, that is, find $T_i$ such that $$ \int_t^{T_i} \phi^{abs}_i(s,t) ds = E_i.$$
        \EndFor
    
         \State \multiline{Compute the time $T$ of the next spike of the system after $t$, and the neuron where the spike occurs by $T \leftarrow \min_{i\in {\bf I}}T_i$, with $j_T \leftarrow \arg\min_{i\in {\bf I}}T_i$}
         \If{$T\leq t_1$} 
         \State append $T$ with mark $j_T$ to ${\bf P}$
         \EndIf
         \State $t\leftarrow T$
    \Until{$t > t_1 $}
    \end{algorithmic} 
    \label{alg:1}
\end{algorithm} 

Note that the quantity $\phi^{abs}_i(s,t)$ can be also seen as the hazard rate of the next potential point $T_i^{(1)}$ after $t$. It is a discrete event approach with the state corresponding to the function $\phi^{abs}_i(.,t)$.

Ogata \cite{Oga81}, inspired by Lewis' algorithm \cite{lewis}, added a thinning (also called rejection) step on top of this procedure because the integral $ \int_t^{T_i^{(1)}} \phi^{abs}_i(s,t) ds$ can be very difficult to compute. To do so (and simplifying a bit), assume that $\phi_i(t|\mathcal{F}^{int}_{t-})$ is upper bounded by a deterministic constant $M$. This means that the point process has always less points than a homogeneous Poisson process with intensity $M$. Therefore Steps 5-6 of Algorithm~\ref{alg:1} can be replaced by the generation of an exponential of parameter $M$, $E'_i$ and deciding whether we accept or reject the point with probability $\phi^{abs}_i(t+E'_i,t)/M$. There are a lot of variants of this procedure: Ogata's original one uses actually the fact that the minimum of exponential variables, is still an exponential variable. Therefore one can propose a next point for the whole system, then accept it for the whole system and then decide on which neuron of the network the event is actually appearing. More details on the multivariate Ogata's algorithm can be found in \cite{Mascart2019}.

As we see here, Ogata's algorithm is very general but clearly needs to simulate the whole system to simulate only one neuron. Moreover starting at time $t_0$, it does not go backward and therefore cannot simulate a Hawkes process in stationary regime. There has been specific algorithms based on clusters representation that aim at perfectly simulate particular univariate Hawkes processes \cite{mr}. The algorithm that we propose here, will also overcome this flaw.

\section{Kalikow decomposition}
\label{Kalikow}
Kalikow decomposition relies on the concept of neighborhood, denoted by $v$, which are picked at random and which gives the portion of time and neuron subsets that we need to look at, to move forward. Typically, for a positive constant $A$, such a $v$ can be:
\begin{itemize}
\item $\{(i, [-A,0))\}$, meaning we are interested only by the spikes of neuron $i$ in the window $[-A,0)$;
\item $\{(i,[-2A,0)),(j,[-2A,-A))\}$, that is, we need the spikes of neuron $i$ in the window $[-2A,0)$ and the spikes of neuron $j$ in the window $[-2A,-A)$;
\item the emptyset $\emptyset$, meaning that we do not need to look at anything to pursue.
\end{itemize}

We need to also define $l(v)$ the total time length of the neighborhood $v$ whatever the neuron is. For instance, in the first case, we find $l(v) = A$, in the second $l(v) = 3 A$ and in the third $l(v)= 0$.

We are only interested by stationary processes, for which the conditional intensity, $\phi_{i}(t \mid \mathcal{F}^{int}_{t^{-}})$, only depends on the intrinsic distance between the previous points and the time $t$ and not on the precise value of $t$ {\it per se}. In this sense the rule to compute the intensity may be only defined at time 0 and then shifted by $t$ to have the conditional intensity at time $t$.  In the same way, the timeline of a neighborhood $v$ is defined as a subset of $\mathds{R}_-^*$ so that information contained in the neighborhood is included in $\mathcal{F}^{int}_{0-}$, and $v$ can be shifted (meaning its timeline is shifted) at position $t$ if need be. We assume that ${\bf I}$ the set of neurons is countable and that we have a countable set of possibilities for the neighborhoods $\mathcal{V}$.

 Then, we say that the process admits {\it a Kalikow decomposition} with bound $M$ and neighborhood family $\mathcal{V}$, if for any neuron $i \in \mathbf{I}$, for all $v \in \mathcal{V}$ there exists a non negative $M$-bounded  quantity $\phi_i^{v}$, which is  $\mathcal{F}^{int}_{0-}$ measurable and whose value only depends on the points  appearing in the neighborhood $v$, and a probability density function $\lambda_{i}(.)$ such that
\begin{equation} \label{Kalikow decomposition}
\phi_{i}(0 \mid \mathcal{F}^{int}_{0^{-}}) = \lambda_i(\emptyset) \phi_i^{\emptyset} + \sum\limits_{v \in \mathcal{V}, v\neq \emptyset}\lambda_i(v) \times \phi_i^{v}
\end{equation}
with $\lambda_i({\emptyset}) + \sum\limits_{v \in \mathcal{V}, v\neq \emptyset} \lambda_i(v) =1.$

Note that because of the stationarity assumptions, the rule to compute the $\phi_i^{v}$'s can be shifted at time $t$, which leads to a predictable function that we call $\phi_i^{v_t}(t)$ which only depends on what is inside $v_t$, which is the neighborhood $v$ shifted by $t$. Note also that $\phi_i^{\emptyset}$, because it depends on what happens in an empty neighborhood, is a pure constant.

The interpretation of \eqref{Kalikow decomposition} is tricky and is not as straightforward as in the discrete case (see \cite{orb}). The best way to understand it is to give the theoretical algorithm for simulating the next event on neuron $i$ after time $t$ (cf. Algorithm~\ref{algth}).

\begin{algorithm}[H]
    \caption{Kalikow theoretical simulation algorithm} 
     \label{algth}
    \begin{algorithmic}[1]
 \Statex \LeftComment{With $[t_0,t_1]$ the interval of simulation for neuron $i$}
 \Statex
    \State Initialize the family of points ${\bf P}=\emptyset$
    
\Statex \LeftComment{   NB: since we are only interested by points on neuron $i$, $j_T=i$ is a useless mark here.} 
\Statex
    \State Initialize $t\leftarrow t_0$
    \Repeat
            \State Draw an exponential variable $E$ with parameter $M$, and compute $T=t+E$.
            \State \multiline{Pick a random neighborhood according to the distribution $\lambda_i(.)$ given in the Kalikow decomposition and shift the neighborhood at time $T$: this is $V_{T}$.}
            \State Draw $X_T$ a Bernoulli variable with parameter $\dfrac{\phi_i^{V_{T}}(T)}{M}$
            \If{$X_T=1$ and $T\leq t_1$} 
               \State append $T$  to ${\bf P}$
            \EndIf
            \State $t\leftarrow T$
    \Until{$t > t_1 $}
    \end{algorithmic}

\end{algorithm} 

This Algorithm is close to Algorithm \ref{algclassic}  but adds a neighborhood choice (Step 5)  with a thinning step (Steps 6-9). 

In Appendix~\ref{appendix:A}, we prove that this algorithm indeed provides a point process with an intensity given by \eqref{Kalikow decomposition} shifted at time $t$.

The previous algorithm cannot be put into practice because the computation of $\phi_i^{V_{T}}$ depends on the points in $V_T$, that are not known at this stage. That is why the efficient algorithm that we propose in the next section goes backward in time before moving forward.

\section{Backward Forward algorithm}

Let us now describe the complete \textit{Backward Forward algorithm} (cf. Algorithm~\ref{algfull}). Note that to do so, the set of points ${\bf P }$ is not reduced, as in the two previous algorithms, to the set of points that we want to simulate but this contains all the points that need to be simulated to perform the task.

\begin{algorithm}[H]
\hspace*{\SpaceReservedForComments}{}%
\begin{minipage}{\dimexpr\linewidth-\SpaceReservedForComments\relax}
    \caption{Backward Forward Algorithm} 
    
    \label{algfull}
    \begin{algorithmic}[1]
        \Statex \LeftComment{With $[t_0,t_1]$ the interval of simulation for neuron $i \in {\bf I}$}
        \Statex
        \State Initialize the family ${\bf V}$ of non empty neighborhoods with $\{(i,[t_0,t_1])\}$
        \State Initialize the family of points ${\bf P}=\emptyset$
        
        \Statex \LeftComment{Each point is a time $T$ with 3 marks: $j_T$ is the neuron on which $T$ appears, $V_T$ for the choice of neighborhood, $X_T$ for the thinning step (accepted/rejected)}
        \State Draw $E$ an exponential variable with parameter $M$ 
        \State Schedule $T_{next}=t_0+E$
        
        \While{$T_{next}<t_1$}
            \State  \multiline{Append to {\bf P}, the point $T_{next}$, 
         with 3 marks: $j_{T_{next}}=i$,  $V_{T_{next}}=$ n.a. and $X_{T_{next}}=$ n.a.  (n.a. stand for {\it not assigned yet})}
            \While{There are points $T$ in {\bf P} with $V_T=$ n.a.}
                \For{ each point $T$ in {\bf P} with $V_T=$ n.a. }
                    \State Update $V_T$ by drawing $V_T$ according to $\lambda_{j_T}$ shifted at time $T$.
                    \If{$V_T\neq \emptyset$}
                        \State \multiline{Find the portion of time/neurons in $V_T$ 
                        which does not intersect the existing non empty neighborhoods in ${\bf V}$}
                        \State Simulate on it a Poisson process with rate $M$ 
                        \State \multiline{Append the simulated points, $T'$, if any, to ${\bf P}$ 
                        with their neuron $j_{T'}$ and with $V_{T'}=X_{T'}=$n.a}
                        \State Append $V_T$ to ${\bf V}$ 
                    \EndIf
                \EndFor
            \EndWhile
            \State Sort the $T$'s in ${\bf P}$ with $X_T=$n.a. in increasing order
            \For{each of them starting with the most backward}
                \State Draw $X_T$ as a Bernoulli variable with parameter $\frac{\phi_{j_T}^{V_T}(T)}{M}$
            \EndFor
            \State Draw $E'$ another exponential variable with parameter $M$ 
            \State $T_{next}\leftarrow T_{next} +E'$
        \EndWhile
        \State The desired points are the points in ${\bf P}$ with marks $j_T=i$, $X_T=1$ and that appear in $[t_0,t_1]$
        \LeftComment{It is possible that the algorithm generated points before $t_0$ and they have to be removed}
  \end{algorithmic} 
    \AddNote[black]{1}{5}{Init.}
    \AddNote[orange]{6}{6}{Init. Marks}
    \AddNote[green]{7}{17}{Backward}
    \AddNote[purple]{18}{25}{Forward}
\end{minipage}
\end{algorithm} 

At the difference with Algorithm \ref{algth}, the main point is that in the backward part we pick at random all the points that may influence the thinning step. The fact that this loop ends comes from the following Proposition.

\begin{proposition} \label{prop 1}
If 
\begin{equation} \label{Sparsity}
\sup_{i\in{\bf I}}\sum\limits_{v \in \mathcal{V}} \lambda_i(v) l(v) M < 1.
\end{equation} 
then the backward part of Algorithm \ref{algfull} ends almost surely in finite time.
\end{proposition}

The proof is postponed to Appendix~\ref{appendix:B}. It is based on branching process arguments. Basically if in Steps 8-16, we produce in average less than one point, either because we picked the empty set in $V_T$ or because the simulation of the Poisson process ended up with a small amount of points, eventually none, then the loop ends almost surely because there is an extinction of the corresponding branching process.

In the backward part, one of the most delicate part consists in being sure that we add new points only if we have not visited this portion of time/neurons before (see Steps 11-13). If we do not make this verification, we may not have the same past depending on the neuron we are looking at and the procedure would not simulate the process we want.

In the forward part, because the backward algorithm stopped just before, we are first sure to have assess all $V_T$'s. Since $\phi_{j}^{V_t}(t)$ is $\mathcal{F}^{int}_{t-}$ measurable, for all $t$, $\phi_{j_T}^{V_T}(T)$ only depends on the points in ${\bf P}$ with mark $X_T=1$ inside $V_T$. The problem in Algorithm \ref{algth}, phrased differently, is that we do not know the marks $X_T$ of the previous points when we have to compute $\phi_{j_T}^{V_T}(T)$.
But in the forward part of Algorithm \ref{algfull}, we are  sure that the most backward point for which the thinning ($X_T=$n.a.) has not taken place, satisfies
\begin{itemize}
\item either $V_T=\emptyset$
\item or $V_T\neq\emptyset$ but either there are no simulated points in the corresponding $V_T$ or the points there come from previous rounds of the loop (Step 5). Therefore their marks $X_T$ have been assigned.
\end{itemize}
Therefore, with the Backward Forward algorithm, and at the difference to Algorithm \ref{algth}, we take the points in an order for which we are sure that we know the previous needed marks.

Figure~\ref{fig:0} describes an example to go step by step through 
Algorithm~\ref{algfull}. The \emph{backward steps} determine all the points that may influence the acceptation/rejection of point $T_{next}$. Notice that whereas usual activity tracking algorithms for point processes~\cite{Mascart2019} automatically detect the active children (influencees), activity tracking in the backward steps detect the parents (influencers). The \emph{forward steps} finally select the points.

\begin{figure}[H]
\captionsetup{type=figure}
\begin{subfigure}{\textwidth}
  \centering
  \includegraphics[width=1\linewidth]{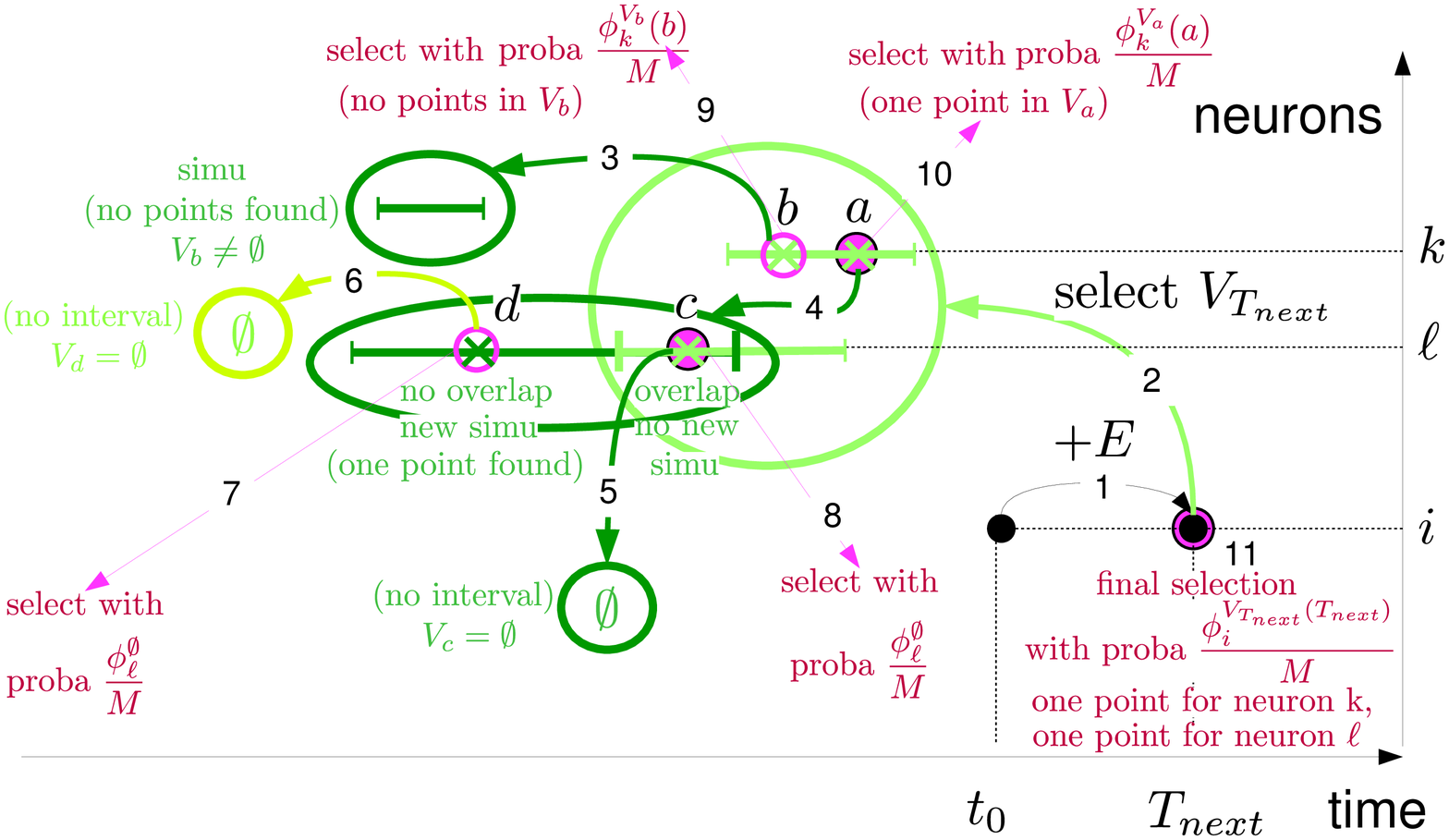}  
  \label{fig:sub-first}
\end{subfigure}
\begin{subfigure}{.5\textwidth}
  \centering
  \includegraphics[width=1\linewidth]{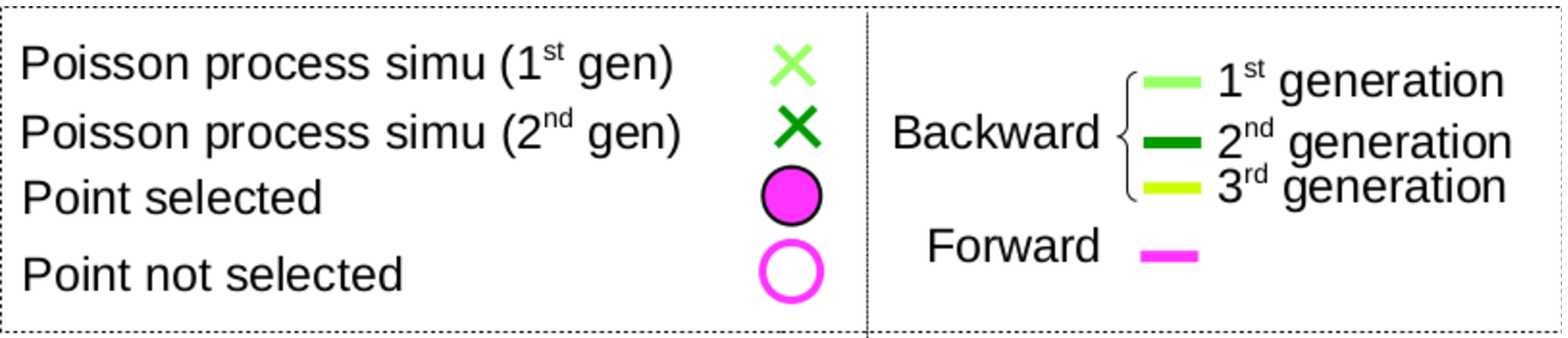}  
  \label{fig:sub-second}
\end{subfigure}
\caption{Main flow example for Algorithm~\ref{algfull}, with backward steps in green (cf. Algorithm~\ref{algfull}, Steps 7-17) and forward steps in purple (cf. Algorithm~\ref{algfull}, Steps 18-25). Following arrow numbers: (1) The next point $T_{next} = t +E$ (cf. Algorithm~\ref{algfull}, Step 4) is scheduled, (2) The neighborhood  $V_{T_{next}}$ is selected 
in the first backward step, a \emph{first generation} of three points ($a,b$ on neuron $k$ and $c$ on neuron $\ell$) is drawn (cf. Algorithm~\ref{algfull}, Step 9), thanks to  a Poisson process, (cf. Algorithm~\ref{algfull}, Steps 11-12) and appended to  ${\bf P}$ (cf. Algorithm~\ref{algfull}, Step 13), (3) at the \emph{second generation}, a non empty neighborhood is found, \textit{i.e.} $V_b \neq \emptyset$ (cf. Algorithm~\ref{algfull}, Steps 9-1), but the  Poisson process simulation does not give any point in it (cf. Algorithm~\ref{algfull}, Step 12), (4) at the \emph{second generation}, the neighborhood $V_a$ is picked, it is not empty and overlap the neighborhood of the first generation (cf. Algorithm~\ref{algfull}, Steps 9-11): therefore there is no new simulation in the overlap ($c$ is kept and belongs to $V_b$ as well as $V_a$) but there is a new simulation thanks to a Poisson process outside of the overlap leading to a new point $d$ (cf. Algorithm~\ref{algfull}, Step 12)
(5) at the \emph{second generation}, for point $c$, one pick the empty neighborhood, 
\textit{i.e.} $V_c = \emptyset$ (cf. Algorithm~\ref{algfull}, Step 9) and therefore we do not simulate any Poisson process, (6) at \emph{third generation}, similarly no point and no interval are generated, \textit{i.e.} $V_d = \emptyset$ (cf. Algorithm~\ref{algfull}, Step 9). This is the end of the backward steps and the beginning of the forward ones, (7) the point $d$ is not selected, acceptation/selection taking place with probability $\frac{\phi_{\ell}^{\emptyset}}{M}$ (cf. Algorithm~\ref{algfull}, Step 20), (8) the point $c$ is  accepted, here again  with probability $\frac{\phi_{\ell}^{\emptyset}}{M}$ (cf. Algorithm~\ref{algfull}, Step 20), (9) the point $b$ is not selected, acceptation taking place, here, with  probability $\frac{\phi_{k}^{V_b}(b)}{M}$ (cf. Algorithm~\ref{algfull}, Step 20), (10) the point $a$ is selected, acceptation taking place, here, with probability $\frac{\phi_{k}^{V_a}(a)}{M}$ (cf. Algorithm~\ref{algfull}, Step 20), (11) The neighborhood of neuron $i$ contains two points, one on neuron $k$ and one on neuron $\ell$ and one selects $T_{next}$ with probability $\frac{\phi_{i}^{V_{T_{next}}}(T_{next})}{M}$}.
\label{fig:0}
\end{figure}








\section{Illustration}

To illustrate in practice the algorithm, we have simulated a Hawkes process as given in \eqref{Hawkesdef}. Indeed such a process has a Kalikow decomposition \eqref{Kalikow decomposition} with bound $M$ and  neighborhood family $\mathcal{V}$ constituted of the $v$'s of the form $v=\{(j,[-A,0))\}$ for some neuron $j$ in ${\bf I}$.
To do that, we need the following choices:
$$\lambda_i(\emptyset) = 1 -  \sum\limits_{j \in {\bf I}} w_{ij} \quad \mbox{ and } \quad \phi_i^{\emptyset} = \dfrac{\nu_i}{\lambda_i(\emptyset)} $$
and for $v$ of the form $v=\{(j,[-A,0))\}$ for some neuron $j$ in ${\bf I}$,
$$\lambda_i(v)=  w_{ij} \quad \mbox{ and } \quad \phi_i^v={\bf Nb}^j_{[-A,0)}\wedge M.$$

We have taken ${\bf I}=\mathds{Z}^2$ and the $w_{ij}$ proportional to a discretized centred symmetric bivariate Gaussian distribution of standard deviation $\sigma$. More precisely, once $\lambda_i(\emptyset)=\lambda_\emptyset$ fixed, picking according to $\lambda_i$ consists in
\begin{itemize}
\item choosing whether $V$ is empty or not with probability $\lambda_\emptyset$
\item if $V\neq\emptyset$, choosing $V=\{(j,[-A,0))\}$ with $j-i=\mbox{round}(W)$ and $W$ obeys a  bivariate $\mathcal{N}(0,\sigma^2)$.
\end{itemize}

In all the sequel, the simulation is made for neuron $i=(0,0)$ with $t_0=0$, $t_1=100$ (see Algorithm \ref{algfull}).
The parameters $M,\lambda_\emptyset$ and $\sigma$ vary. The parameters $\nu_i=\nu$ and $A$ are fixed accordingly by 
$$\nu=0.9M\lambda_\emptyset \mbox{ and } A=0.9 M^{-1}(1-\lambda_\emptyset)^{-1},$$
to ensure that $\phi_i^{\emptyset}<M$ and \eqref{Sparsity}, which amounts here to $(1-\lambda_\emptyset)AM<1$.

On Figure \ref{fig:1}(a), with $M=2$, $\sigma = 1$ and $\lambda_{\emptyset}$ small, we see the overall spread of the algorithm around the neuron to simulate (here $(0,0)$). Because we chose a Gaussian variable with small variance for the $\lambda_i$'s, the spread is weak and the neurons very close to the neuron to simulate are requested a lot of time at Steps 9-11 of the algorithm. This is also where the algorithm spent the biggest amount of time to simulate Poisson processes. Note also that roughly to simulate 80 points, we need to simulate 10 times more points globally in the infinite network. Remark also on Figure \ref{fig:1}(b), the avalanche phenomenon, typical of Hawkes processes: for instance the small cluster of black points on neuron with label 0 (i.e. (0,0)) around time 22, is likely to be due to an excitation coming for the spikes generated (and accepted) on neuron labeled 8 and self excitation. The beauty of the algorithm is that we do not need to have the whole time line of neuron 8 to trigger neuron 0, but only the small blue pieces: we just request them at random, depending on the Kalikow decomposition.

On Figure \ref{fig:2}, we can first observe that when the parameter which governs the range of $\lambda_i$'s increase, the global spread of the algorithm increase. In particular, comparing the top left of Figure  \ref{fig:2} to Figure \ref{fig:1} where the only parameter that changes is $\sigma$, we see that the algorithm is going much further away and simulates much more points for a sensible equal number of points to generate (and accept) on neuron (0,0). Moreover we can observe that
\begin{itemize}
\item From  left  to  right,  by  increasing $\lambda_\emptyset$,  it  is  more  likely  to  pick  an  empty neighborhood and as a consequence, the spread of the algorithm is smaller. By increasing $\nu=0.9M\lambda_\emptyset$, this also increases the total number of points produced on neuron (0;0).
\item From  top  to  bottom,  by  increasing $M$,  there  are  more  points  which  are simulated in the Poisson processes (Step 12 of Algorithm 3) and there is also a stronger interaction (we do not truncate that much the number of points in $\phi^v$). Therefore,  the spread becomes larger and more uniform too, because there are globally more points that are making requests. Moreover, by having a basic rate $M$ which is 10 times bigger, we have to simulate roughly 10 times more points.
\end{itemize}

\begin{figure}[H]
\captionsetup{type=figure}
\begin{subfigure}{.4\textwidth}
  \centering
  \includegraphics[width=1\linewidth]{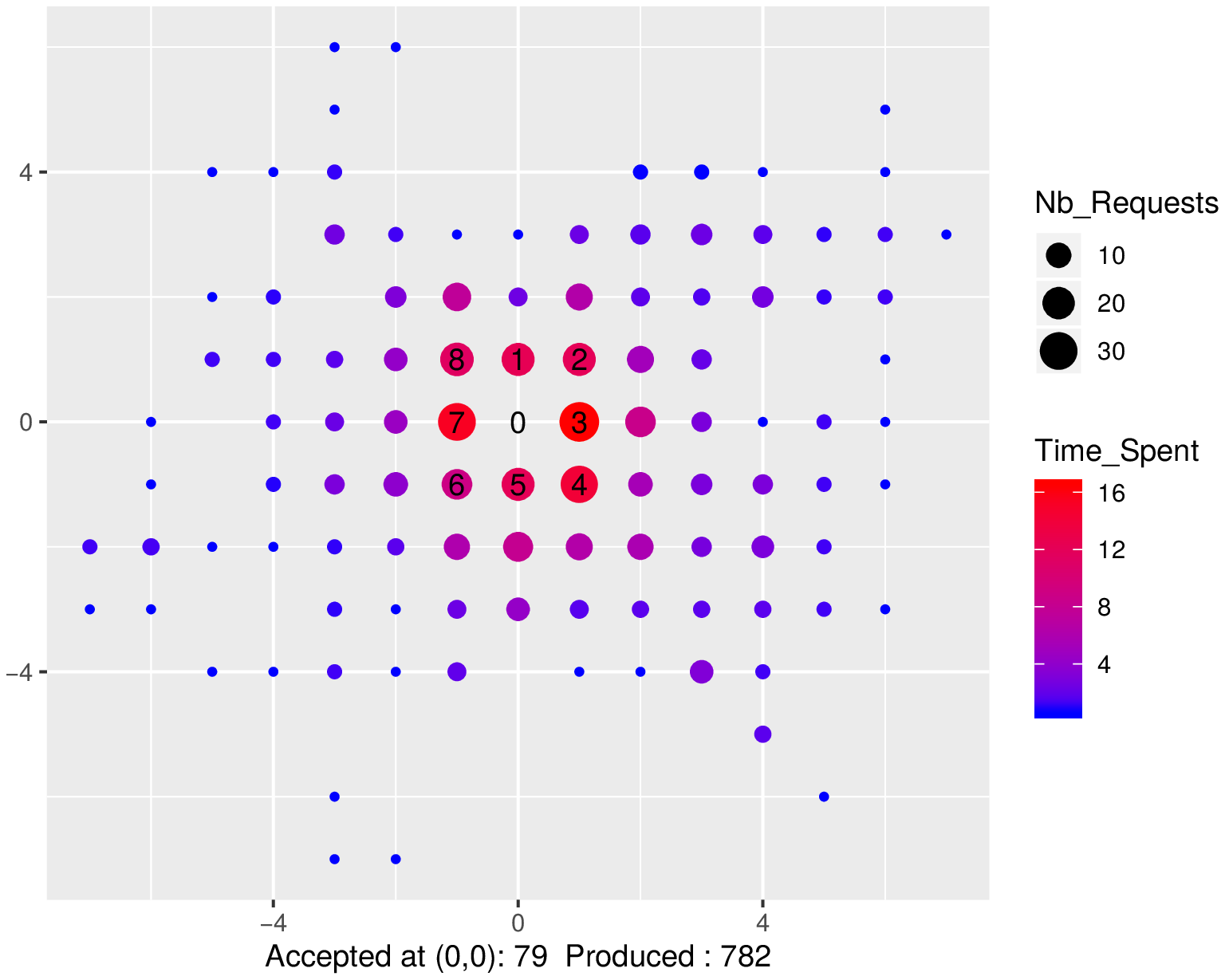}  
  \caption{Summary of one simulation}
  \label{fig:sub-first}
\end{subfigure}
\begin{subfigure}{.4\textwidth}
  \centering
  \includegraphics[width=1\linewidth]{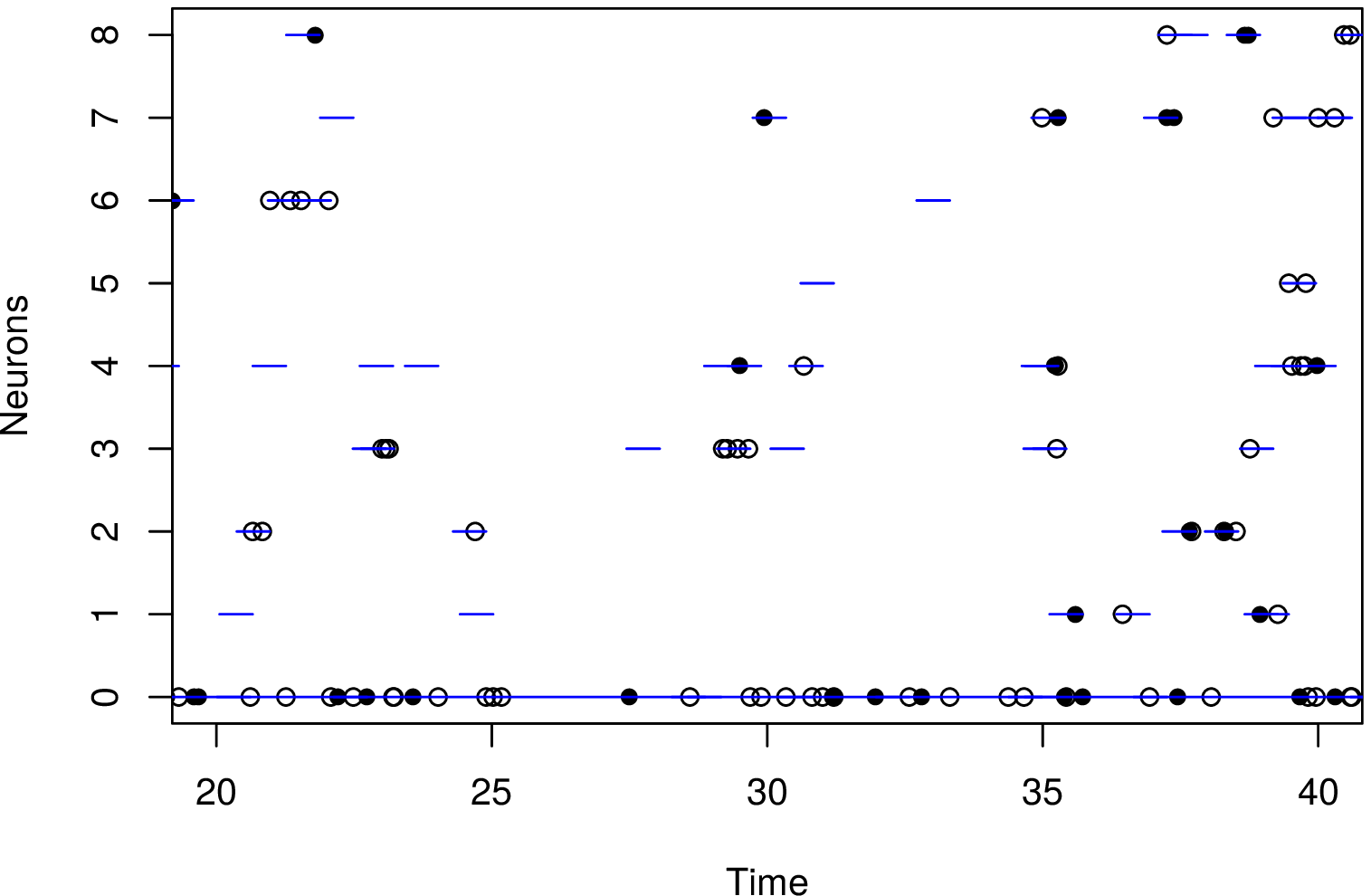}  
  \caption{Extract of the time simulation for neurons 0 to 9}
  \label{fig:sub-second}
\end{subfigure}
\caption{Simulation for $M=2$, $\sigma =1$, $\lambda_{\emptyset} = 0.25$. For each neuron in $\mathds{Z}^2$, that have been requested in Steps 9:11, except the neuron of interest $(0,0)$, have been counted the total number of requests, that is the number of time a $V_T$ pointed towards this neuron (Steps 9 and 11) and the total time spent at this precise neuron simulating a homogeneous Poisson process (Step 12). Note that since the simulation is on $[0,100]$ the time spent at position $(0,0)$ is at least 100. On (a), the summary for one simulation with below the plot, number of points accepted at neuron $(0,0)$ and total number of points that have been simulated. Also annotated on (a), with labels between 0 and 8, the 9 neurons for which the same simulation in [20,40] is represented in more details on (b). More precisely on (b), in abscissa is time and the neuron labels in ordinate. A plain dot represents an accepted point, by the thinning step (Step 20 of Algorithm 3), and  an empty round, a rejected point. The blue pieces of line represent the non empty neighborhoods that are in  ${\bf V}$.
}
\label{fig:1}
\end{figure}

\begin{figure}[H]
\captionsetup{type=figure}
\centering
\begin{tabular}{cc}
\includegraphics[width=.5\linewidth]{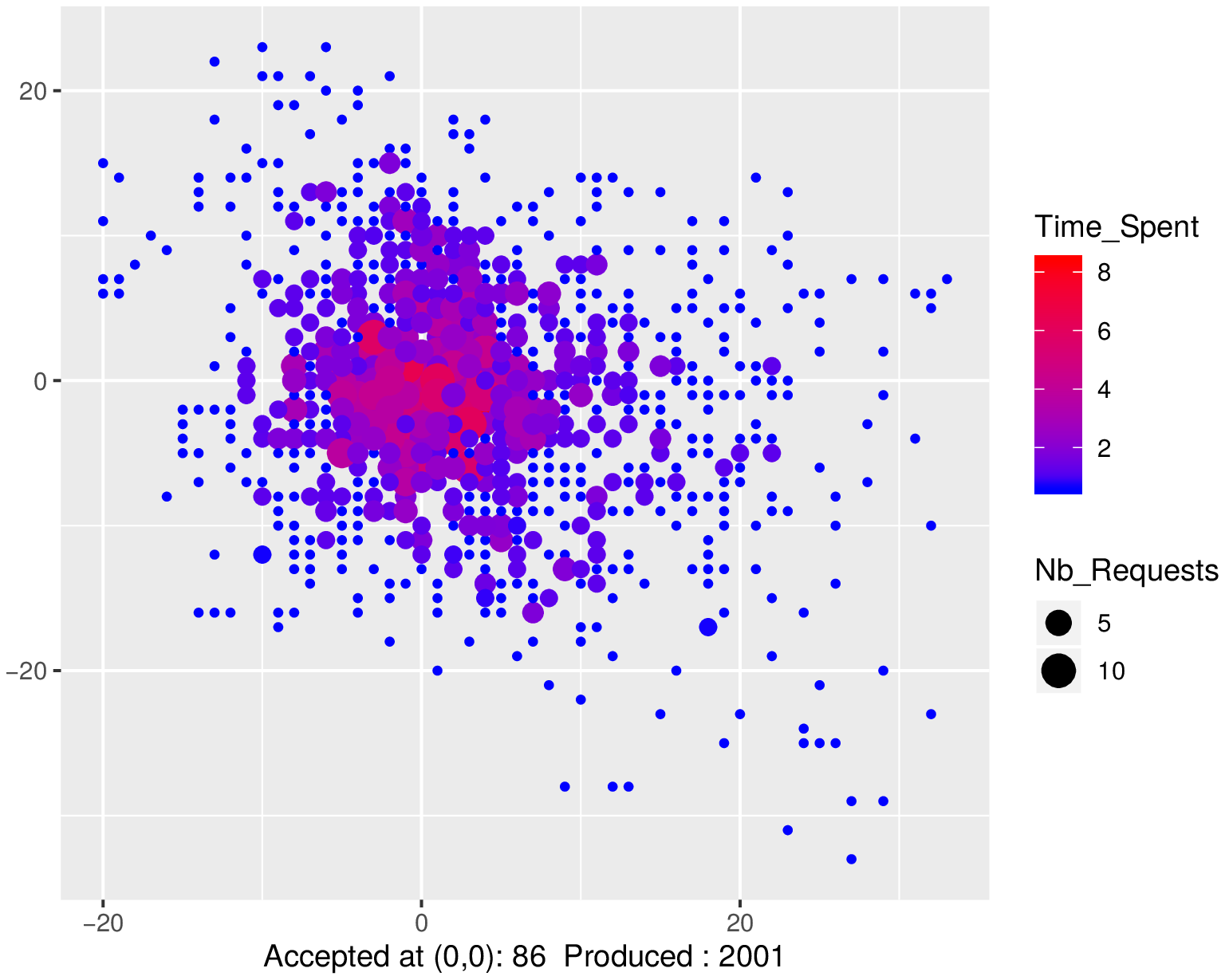}& \includegraphics[width=.5\linewidth]{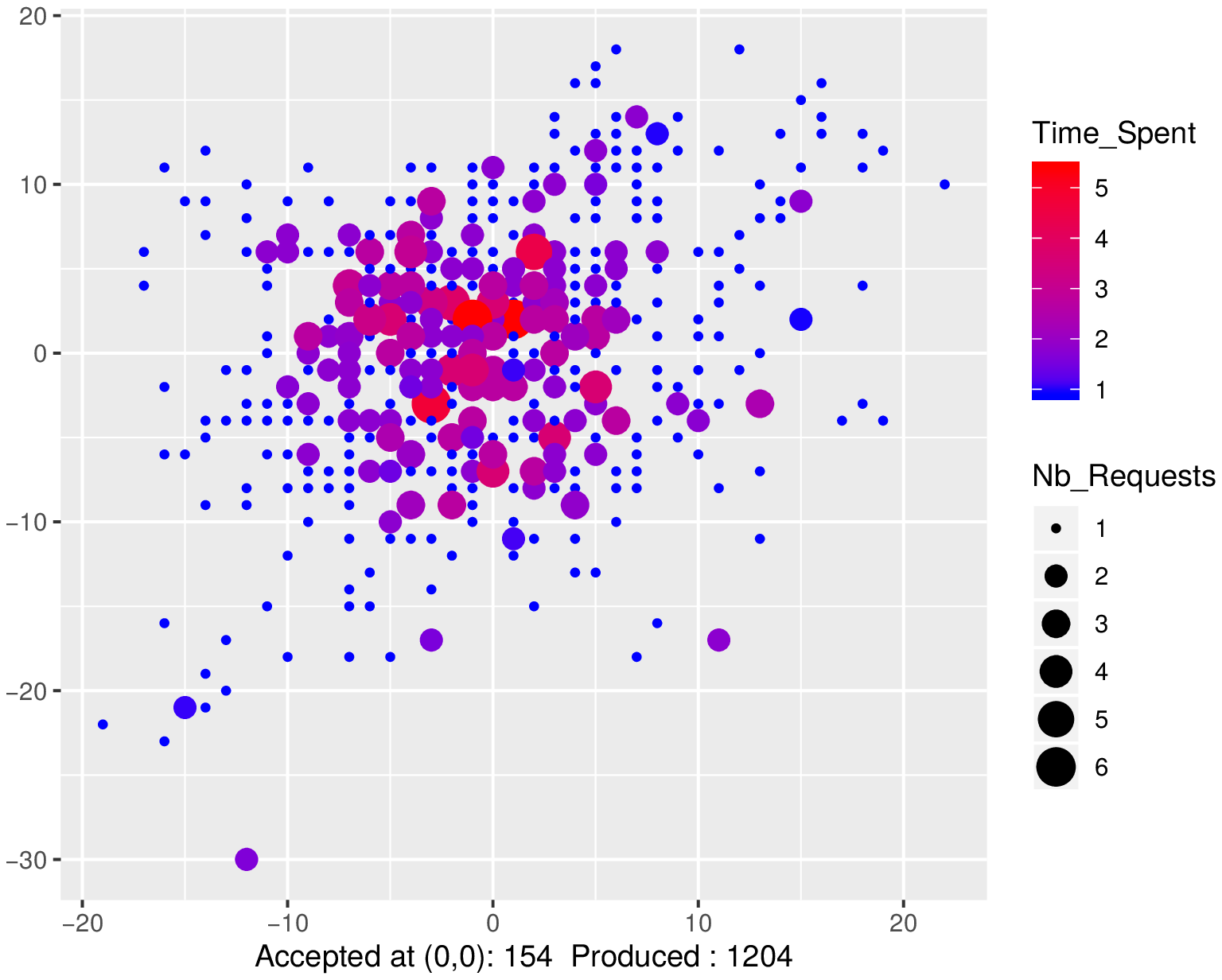}  \\
\includegraphics[width=.5\linewidth]{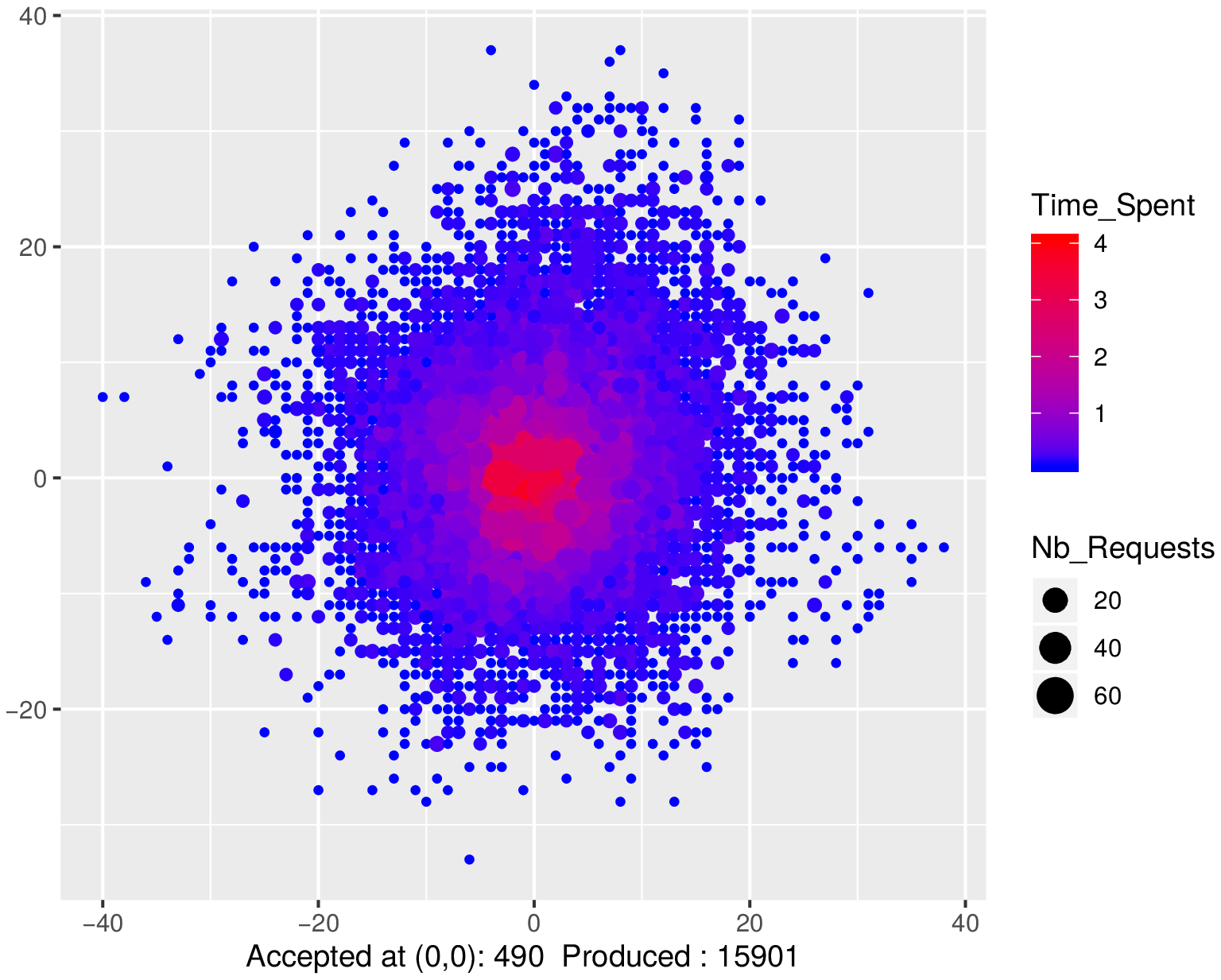}& \includegraphics[width=.5\linewidth]{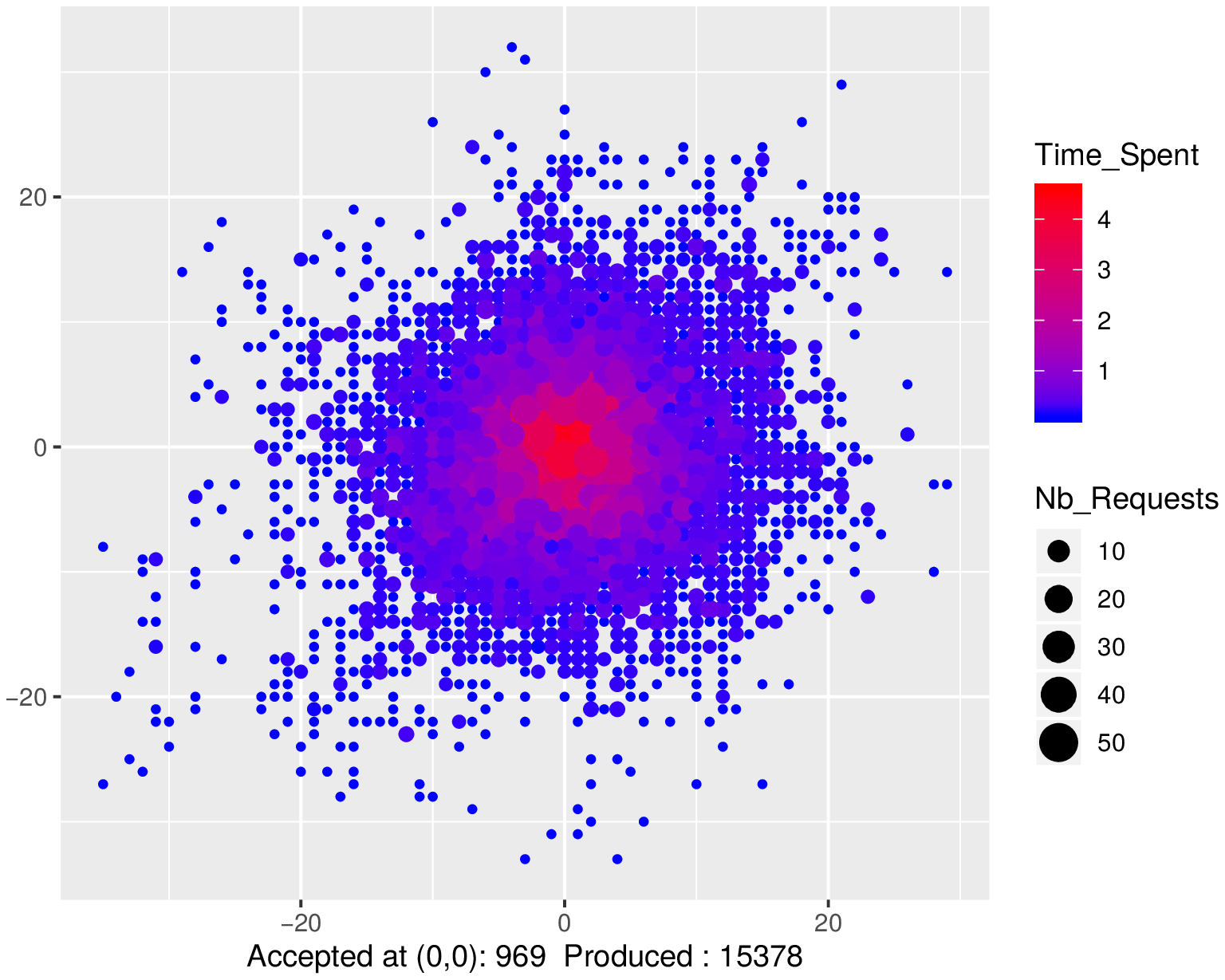} 
\end{tabular}
\caption{Simulation for 4 other sets of parameters, all of them with $\sigma=3$. Summaries as explained in Figure \ref{fig:1}. On top, $M=2$; on bottom, $M=20$. On the left part, $\lambda_{\emptyset} =0.25$, on the right part, $\lambda_{\emptyset}=0.5.$  }
\label{fig:2}
\end{figure}

\section{Conclusion}
We derived a new algorithm for simulating the behavior of one neuron embedded in an infinite network. This is possible thanks to the Kalikow decomposition which allows picking at random the influencing neurons. As seen in the last section, it 
is computationnally tractable in practice to simulate open systems in the physical sense. A question that remains open for future work is whether we can prove that such a decomposition exists for a wide variety of processes, as it has been shown 
in discrete time (see \cite{gl1,gl2,orb}).

\section*{Acknowledgements}
This work was supported by the French government, through the UCA$^{Jedi}$ Investissements d'Avenir managed by the National Research Agency (ANR-15-IDEX-01) and by the interdisciplinary Institute for Modeling in Neuroscience and Cognition (NeuroMod) of the Universit\'e C\^ote d'Azur. The authors would like to thank Professor E.L{\"o}cherbach from Paris 1 for great discussions about Kalikow decomposition and Forward Backward Algorithm.

\appendix

\section{Link between Algorithm 2 and the Kalikow decomposition}
\label{appendix:A}

To prove that Algorithm \ref{algth} returns the desired processes, let us use some additional and  more mathematical notation. Note that all the points simulated on neuron $i$ before being accepted or not can be seen as coming from a common Poisson process of intensity $M$, denoted $\Pi_i$. 
For any $i \in \mathbf{I}$, we denote 
the  arrival times of $\Pi_i$, $(T^i_n)_{n \in \mathds{Z}}$, with $T^i_1$ being the first positive time.

As in Step 6 of Algorithm \ref{algth}, we attach to each point of $\Pi^i$ a stochastic mark $X$ given by,
\begin{equation}
X^i_{n} = 
\begin{cases}
1 \quad &\text{if} \quad T^i_n \, \text{is accepted in the thinning procedure}\\
0  \quad  &otherwise.
\end{cases}
\end{equation}
Let us also define $V^i_n$ the neighborhood choice of $T^i_n$ picked at random and independently of anything else according to $\lambda_i$ and shifted at time $T^i_n$.

In addition, for any $i \in \mathbf{I}$, define $N^i= (T^i_n,X^i_n)_{n \in \mathds{Z}}$ an $E$-marked point process with $E=\{0;1\}$.
In particular, following the notation in Chapter VIII of \cite{Bre81}, for any $i \in \mathbf{I}$, let
\begin{align*}
N^i_t(mark) &= \sum_{n \in \mathds{Z}} \mathds{1}_{X^i_n =mark} \mathds{1}_{T^i_n \leq t} \quad \text{for} \quad mark \in E \\
\mathcal{F}^{N}_{t^{-}}&= \bigvee_{i \in \mathbf{I}}  \sigma(N^i_s({0}), N^i_s({1}); s<t) \quad \mbox{and} \quad
\mathcal{F}^{N(1)}_{t^{-}}&= \bigvee_{i \in \mathbf{I}}  \sigma(N^i_s({1}); s<t).
\end{align*}

Moreover note that $(N_t^i(1))_{t\in \mathds{R}}$ is the counting process associated to the point processs ${\bf P}$ simulated by Algorithm \ref{algth}. Let us denote by $\varphi_i(t)$, the formula given by \eqref{Kalikow decomposition} and shifted at time $t$. Note that since the $\phi_i^v$'s are $\mathcal{F}_{0-}^{int}=\mathcal{F}_{0-}^{N(1)}$, $\varphi_i(t)$ is  $\mathcal{F}^{N(1)}_{t^{-}}$ measurable. We also denote $\varphi_i^v(t)$ the formula of $\phi_i^v$ shifted at time $t$.

With this notation, we can prove the following.

\begin{proposition}
The  process $(N^i_t(1))_{t\in \mathds{R}}$ admits $\varphi_i(t)$  as  $\mathcal{F}^{N(1)}_{t^{-}}$-predictable intensity.
\end{proposition}

\begin{proof}
Following the technique in Chapter 2 of \cite{Bre81}, let us take $C_t$ a non negative predictable function with respect to (w.r.t) $\mathcal{F}^{N^i(1)}_{t}$ that
is $\mathcal{F}^{N(1)}_{t-}$ measurable and therefore $\mathcal{F}^{N}_{t-}$ measurable . We have, for any $i \in \mathbf{I}$,
$$
\mathds{E}\left( \int\limits_{0}^{\infty} C_t dN^i_t(1) \right) = \sum_{n=1}^\infty \mathds{E}\left(C_{T^i_n} \mathds{1}_{X^i_{n} = 1} \right)$$
Note that by Theorem T35 at Appendix A1 of \cite{Bre81}, any point $T$ should be understood as a stopping time, and  that by Theorem T30 at Appendix A2 of \cite{Bre81}, 
$$\mathcal{F}^{N}_{T-} = \bigvee_j \sigma\{T^j_m, X^j_m \, \text{such that} \, T^j_m < T\}$$ 
So
$$
\mathds{E}\left( \int\limits_{0}^{\infty} C_t dN^i_t(1) \right) = \sum_{n=1}^\infty \mathds{E}\left(C_{T^i_n} \mathds{E}( \mathds{1}_{X^i_{n} = 1} | \mathcal{F}^{N}_{{T^i_n}^{-}}, {V}^i_n)\right)= \sum_{n=1}^\infty \mathds{E} \left(  C_{T^i_n}  \dfrac{\varphi_i^{V^i_n}(T^i_n)}{M} \right).$$
Let us now integrate with respect to the choice $V^i_n$, which is independent of anything else.
$$
\mathds{E}\left( \int\limits_{0}^{\infty} C_t dN^i_t(1) \right) = \sum_{n=1}^\infty \mathds{E}\left(C_{T^i_n} \frac{\lambda_i(\emptyset) \varphi_i^{\emptyset} + 
\sum\limits_{v \in \mathcal{V}, v\neq\emptyset}\lambda_i(v) \times \varphi_i^{v}(T^i_n)}{M} \right)= \mathds{E}\left(\int\limits_0^{\infty} C_t \frac{\varphi_i(t)}{M} d\Pi^i(t)\right).$$
Since $\Pi^i$ is a Poisson process with respect to $(\mathcal{F}_t^N)_t$
with intensity $M$, and since  $C_t \frac{\varphi_i(t)}{M}$ is $\mathcal{F}_{t-}^N$ measurable, we finally have that
$$
\mathds{E}\left( \int\limits_{0}^{\infty} C_t dN^i_t(1) \right) = \mathds{E}\left(\int\limits_0^{\infty}C_t\varphi_i(t) dt\right),$$
which ends the proof.
\end{proof}

\section{Proof of Proposition \ref{prop 1}}
\label{appendix:B}

\begin{proof}
We do the proof for the \textit{backward} part, starting with $T=T_{next}$ as the next point after $t_0$ (Step 4 of Algorithm \ref{algfull}), the proof being similar for the other $T_{next}$ generated at Step 23. We construct a tree with root $(i,T)$. For each point $(j_{T'},T')$ in the tree, the points which are simulated in $V_{T'}$ (Step 12 of Algorithm \ref{algfull}) define the children of $(j_{T'}, T')$ in the tree. This forms the tree $\tilde{\mathcal{T}}$.

Let us now build a tree $\tilde{\mathcal{C}}$ with root $(i,T)$ (that includes the previous tree) by mimicking the previous procedure in the \textit{backward} part, except that we  simulate on the whole neighborhood even if it has a part that intersects with previous neighborhoods (if they exist) (Step 11-12 of Algorithm \ref{algfull}). By doing so, we make the number of children at each node independent of anything else.

If the tree $\tilde{\mathcal{C}}$ goes extinct then so does the tree  $\tilde{\mathcal{T}}$ and the backward part of the algorithm terminates.

But if one only counts the number of children in the tree $\tilde{\mathcal{C}}$, we have a marked branching process whose reproduction distribution for the mark $i$ is given by
\begin{itemize}
\item no children with probability $\lambda_i(\emptyset)$
\item Poissonian number of children with parameter $l(v)M$ if $v$ is the chosen neighborhood with probability $\lambda_i(v)$
\end{itemize}
This gives that the average number of children issued from a node with the mark $i$ is
$$\zeta_i=\lambda_i(\emptyset) \times 0+\sum_{v\in\mathcal{V}, v \neq \emptyset} \lambda_i(v)  l(v)M.$$

If we denote $\tilde{\mathcal{C}}^k$ as the collection of points in the tree $\tilde{\mathcal{C}}$ at generation $k$, and by $K_{T'}$ the set of points generated independently as a Poisson process of rate $M$ inside $V_{T'}$, we see recursively that  
\begin{equation*}
 \tilde{\mathcal{C}}^{k+1}  
 = \bigcup_{T'\in \tilde{\mathcal{C}}^{k}} K_{T'}
\end{equation*}
But $$\mathds{E}(|K_{T'}| | T')= \zeta_{j_{T'}}.$$
Therefore, if we denote the total number of sites in $\tilde{\mathcal{C}}^k$ by $Z^{(k)}$, we have
$$\mathds{E}(Z^{(k+1)}|\tilde{\mathcal{C}}^k) \leq  Z^{(k)} \sup_{i\in I} \zeta_i.$$

One can then conclude by recursion that,
\begin{equation*}
\mathds{E}(Z^{(k)}) \leq (\sup_{i\in I} \zeta_i)^{k} <1.
\end{equation*}
The last inequality use the sparsity neighborhood assumption. Then we deduce that, the mean number of children in each generation goes to $0$ as $k$ tends to infinity. So by using classical branching techniques in \cite{Mel10}, we conclude that the tree $\tilde{\mathcal{C}}$ will go extinct almost surely. This  also implies that, the backward steps end a.s. 
\end{proof}

\end{document}